\documentclass[12pt,draftclsnofoot,journal,onecolumn]{IEEEtran}

\usepackage{cite}
\usepackage{graphicx}
\usepackage{amsmath}
\usepackage{amsfonts}
\usepackage{array}
\usepackage{booktabs}
\usepackage{amssymb}
\usepackage{subfigure}
\newtheorem{lemma}{Lemma}
\newtheorem{theorem}{Theorem}

\newtheorem{remark}{Remark}
\newtheorem{AS}{Assumption}

\begin{document}

\title{Adaptive Spatial Intercell Interference Cancellation in Multicell Wireless Networks}
\author{Jun Zhang and Jeffrey G. Andrews
\thanks{The authors are with the Wireless Networking and Communications Group (WNCG), Department of Electrical and Computer Engineering, The University of Texas at Austin. Email: jzhang06@mail.utexas.edu, jandrews@ece.utexas.edu.}}

\maketitle

\begin{abstract}
Downlink spatial intercell interference cancellation (ICIC) is considered for mitigating other-cell interference using multiple transmit antennas. A principle question we explore is whether it is better to do ICIC or simply standard single-cell beamforming. We explore this question analytically and show that beamforming is preferred for all users when the edge SNR (signal-to-noise ratio) is low ($<0$ dB), and ICIC is preferred when the edge SNR is high ($>10$ dB), for example in an urban setting. At medium SNR, a proposed adaptive strategy, where multiple base stations jointly select transmission strategies based on the user location, outperforms both while requiring a lower feedback rate than the pure ICIC approach. The employed metric is sum rate, which is normally a dubious metric for cellular systems, but surprisingly we show that even with this reward function the adaptive strategy also improves fairness. When the channel information is provided by limited feedback, the impact of the induced quantization error is also investigated. It is shown that ICIC with well-designed feedback strategies still provides significant throughput gain.
\end{abstract}


\section{Introduction}
The performance of contemporary multicell wireless networks is limited by other-cell interference (OCI), due to cochannel transmission in other cells. This performance degradation is especially severe for users close to the cell edge. MIMO (Multiple-Input Multiple-Output) transmission theoretically provides significant throughput gain, but the OCI is an even more complex obstacle due to the increased number of interfering sources \cite{Catreux00,Blum03,AndChoHea06}. Conventional approaches to mitigate OCI include static frequency reuse, sectoring, and spread spectrum, which require little coordination among neighboring base stations (BSs). More recently, \emph{multicell processing}, or \emph{BS coordination}, has been proposed as a more efficient way to suppress OCI through coordination among multiple BSs \cite{Shamai04}. In this paper, we propose a novel adaptive multicell interference suppression technique.

In cellular MIMO networks, due to the lack of cooperation amongst mobile users, downlink transmission is usually more difficult than the uplink, and is often the capacity-limiting link. Therefore, this paper focuses on the downlink. Multicell processing in the downlink can be categorized into two classes:

\textbf{1. Coordinated single-cell transmission}: data is transmitted from a single BS, and the OCI suppression is achieved through joint resource allocation among multiple BSs, such as joint power control and user scheduling \cite{GesKia07}. Neighboring cells share such information as the offer load in each cell, the channel state information (CSI) of edge users, but no data exchange is required.

\textbf{2. Coordinated multicell transmission}: in addition to the information shared in coordinated single-cell transmission, BSs need to exchange user data. A central unit (CU) is normally needed for joint processing of data transmission for BSs that join the coordination, so each user receives data from multiple BSs. Ideally, assuming full CSI and all the data available at the CU, coordinated multicell transmission is able to eliminate all the OCI and the system is no longer interference-limited \cite{Shamai04,SomZai07IT}.

Although coordinated multicell transmission is able to provide a considerable performance gain through efficiently exploiting the available spatial degrees of freedom, it requires a significant amount of inter-BS information exchange and is of high complexity. This would be quite challenging for practical implementation. First, the large overhead and information exchange would put onerous demands on backhaul capacity; second, precise synchronization among different
BSs is required; third, the CSI from each mobile user is required at all the coordinated BSs, which makes CSI estimation and feedback daunting.

On the other hand, coordinated single-cell transmission is of lower overhead and complexity, as no inter-BS data exchange is required, and normally each user needs to provide instantaneous or statistical CSI only to some of its neighboring BSs. In this paper, we consider a multicell network with multiple antennas at each BS. Coordinated single-cell transmission is applied in the form of intercell interference cancellation (ICIC) through zero-forcing (ZF) precoding. Canceling OCI for neighboring cells consumes available spatial degrees of freedom, so it reduces the received signal power for the home user, and is not necessarily optimal at each BS. We propose an adaptive ICIC strategy where multiple BSs jointly select transmission techniques based on user locations. Each BS only needs to exchange the location of its home user with neighboring BSs, and the CSI of users in neighboring cells is required only when ICIC is applied.

\subsection{Related Work}
Coordinated multicell transmission, also called \emph{network MIMO}, has recently drawn significant attention. In a network MIMO system, multiple coordinated BSs effectively form a ``super BS'', which transforms an interference channel into a MIMO broadcast channel, with a per-BS power constraint \cite{Boccardi06,KarYat07ISIT,YuLan07Tsp}. The optimal dirty paper coding (DPC) \cite{Costa83,WeiSte06IT} and sub-optimal linear precoders have been developed for network MIMO \cite{ShamaiVTC01,Zhang04,JafFos04,Karakayali06a,FosKar06IEE,ZhaChe09Twc}. With simplified network models, analytical results have appeared in \cite{HuaVen2004Asi,Somekh06,AktBac06IT,JinTse08EURASIP}.

In practice, the major challenges for network MIMO concern complexity and overhead. For example, the requirement for CSI grows in proportion to the number of BS antennas, the number of BSs, and the number of users. The complexity of joint processing also grows with the network size. To limit the complexity and CSI requirements, cluster-based coordination is one approach \cite{Ven07PIMRC,BocHua07PIMRC,PapGes08ICC,ZhaChe09Twc}. To reduce the complexity, distributed decoding and beamforming for network MIMO systems were proposed in \cite{AktEva08TWC,NgEva04Glob,NgEva08IT}. In \cite{GesHjo06,SkjGes06ICASSP}, BS coordination with hybrid channel knowledge was investigated, where each BS has full information of its own CSI and statistical information of other BSs' channels. Limited backhaul capacity \cite{MarFet07EW,SanSo07ISIT} and synchronization \cite{ZhaMeh08TWC,JunWir08} have also been treated to some extent. A WiMAX based implementation of network MIMO was done in \cite{VenHua09EURASIP}, for both uplink and downlink in the indoor environment.

Coordinated single-cell transmission, where the traffic data for each user comes from a single BS, is of lower complexity, requires less inter-BS information exchange, and has lower CSI requirements. Intercell scheduling has been shown to be able to expand multiuser diversity gain versus static frequency planning \cite{ChoAnd08Twc}, while coordinated load balancing and intercell scheduling were investigated in \cite{DasVis03INFOCOM,SanWan08}. Multi-cell power control algorithms were proposed in \cite{KiaGes07WCNC,KiaGes08Twc}. The use of multiple antennas to suppress OCI has also been investigated as a coordinated single-cell transmission strategy, mainly in the form of receive combining. Optimal signal combining for space diversity reception with cochannel interference in cellular networks was proposed in \cite{Win84Tvt,WinSal94Tcomm}. In \cite{JinAnd09ICC,HuaAnd08IT}, spatial interference cancellation with multiple receive antennas has been exploited in ad hoc networks, which bear some similarity to multicell networks. Receive combining, however, can be applicable mainly in the uplink, as there are usually multiple antennas at the BS but only a small number of antennas at the mobile. Downlink beamforming in multicell scenarios was investigated in \cite{FarLiu98JSAC,DahYu08CISS}, with the objective of minimizing the transmit power to support required receive SINR constraints at mobiles.

\subsection{Contributions}
In this paper, we investigate spatial ICIC using ZF precoding to suppress downlink OCI and improve the system throughput. The main contributions are summarized as follows.

\textbf{Throughput analysis and adaptive ICIC:} We provide closed-form expressions for the ergodic achievable sum rates when BSs take different transmission strategies, including selfish beamforming and doing ICIC for some of the neighboring cells. Adaptive ICIC is proposed to maximize the sum throughput by jointly selecting the transmission strategy at each BS based on user locations.

\textbf{Strategy selection:} It is shown that when the edge SNR is high, each BS tends to do ICIC for neighboring cells; when the edge SNR is low, each BS tends to do beamforming for its own user without ICIC; for medium edge SNR, the proposed adaptive strategy improves the sum and edge throughput and also reduces the required CSI compared to static ICIC. Numerical results show that in a 3-cell network the average throughput is increased by about half while the edge throughput is increased three-fold when the average edge SNR is $15$ dB. In addition, with the sum throughput as the performance metric, the BS with a cell interior user is willing to help the edge user in the neighboring cell, i.e. it encourages fairness.

\textbf{Impact of limited feedback:} If the CSI at each BS is obtained through limited feedback, the induced quantization error will degrade the performance of ICIC. We provide accurate approximations for the achievable throughput with limited feedback. It is shown that to keep a constant rate loss versus perfect CSI, the number of feedback bits to the neighboring helper BS needs to grow linearly with both the number of transmit antennas and the edge SNR (in dB). With a constraint on the total number of feedback bits, the performance can be improved by adaptively allocating the available feedback bits.

The rest of the paper is organized as follows. The system model is presented in Section \ref{Sec:Model}, together with the proposed transmission strategy. Adaptive ICIC in a 2-cell network is investigated in Section \ref{Sec:2cell}, while the extension to 3-cell and general multicell networks is in Section \ref{Sec:3cell}. The impact of limited feedback on the ICIC system is investigated in Section \ref{Sec:LFB}. Numerical results are provided in Section \ref{Sec:Num} and conclusions are made in Section \ref{Sec:Con}.

\section{System Model}\label{Sec:Model}
We consider a multicell wireless network, where each BS has $N_t$
antennas and each mobile user has a single antenna. Each mobile is associated with a \emph{home BS}, which is the closest one. Universal frequency reuse is assumed. An active mobile, i.e. the one being scheduled for transmission, receives a data signal from its home BS while suffering OCI from other BSs. ICIC in the spatial domain using multiple antennas is applied to suppress OCI. The BS applying ICIC for a user is called its \emph{helper BS}. A 2-cell network
is shown in Fig. \ref{fig:2cell}, which will be used as an instructive example in this section. We consider the downlink transmission, i.e. from the BS to mobiles. Following are some assumptions we make in our study.
\begin{AS}
\emph{The neighboring BSs can exchange the location and CSI of each active user, but may not share traffic data.}
\end{AS}

With this assumption, each BS is able to do ICIC for its neighboring cells, but coordinated multicell transmission cannot be performed.

\begin{AS}
\emph{There is one active user served in each cell at each time
slot with precoding at the BS.}
\end{AS}

Denote the user and the BS in the $i$-th cell as the $i$-th user and the $i$-th BS, i.e. only a single user is active per BS per time slot, which precludes multi-user MIMO (MU-MIMO). The results could be extended to MU-MIMO in future work. With multi-antenna transmission at each BS, it is difficult to measure the interference from neighboring cells, which depends strongly on the active precoder, so we do not consider channel-dependent scheduling in the current work.

\subsection{Adaptive Coordination}
With multiple antennas, although each BS is able to do ICIC to cancel OCI for neighboring cells, this may be suboptimal, as ICIC will reduce the received signal power for its own user. For a 2-cell network, we assume each BS can select one of two strategies:
\begin{enumerate}
\item \textbf{Selfish beamforming}: it serves its own user with eigen-beamforming and does not cancel interference for the other cell. This strategy is denoted as $BF$.
\item \textbf{Interference cancellation}: it does interference cancellation for some of the neighboring cells. Denote $IC(\mathcal{I}_i)$ as the strategy that the $i$-th BS is doing ICIC for the users with indices in the set $\mathcal{I}_i$. In a 2-cell network, $IC(\mathcal{I}_i)$ is simplified as $IC$ without ambiguity.
\end{enumerate}
So the strategy set is
$\mathcal{S}_1=\{BF,IC\}$ and $\mathcal{S}_2=\{BF,IC\}$ for BS 1 and BS 2, respectively, and the
strategy pair taken by 2 BSs is
$(s_1,s_2)\in\mathcal{S}_1\times\mathcal{S}_2$, where
$\mathcal{S}_1\times\mathcal{S}_2$ is the Cartesian product.

When the active strategy pair is $(s_1,s_2)$, the two received signals are given as
\begin{align}
y_1(s_1,s_2)&=\sqrt{P^r_{1,1}}\mathbf{h}_{1,1}^*\mathbf{f}_{1,s_1}x_1+\sqrt{P^r_{1,2}}\mathbf{h}_{1,2}^*\mathbf{f}_{2,s_2}x_2+z_1,\label{eq:1stuser}\\
y_2(s_1,s_2)&=\sqrt{P^r_{2,2}}\mathbf{h}_{2,2}^*\mathbf{f}_{2,s_2}x_2+\sqrt{P^r_{2,1}}\mathbf{h}_{2,1}^*\mathbf{f}_{1,s_1}x_1+z_2,\label{eq:2nduser}
\end{align}
where $\mathbf{a}^*$ is the conjugate transpose of a vector $\mathbf{a}$ and
\begin{itemize}
\item $P^r_{i,j}$ is the received power at the $i$-th user from the
$j$-th BS. We use the path loss model
$P^r_{i,j}=P_0\left(D_0/d_{i,j}\right)^\alpha$, where $P_0$ is the
received signal power at the reference distance $D_0$, and $d_{i,j}$
is the distance between the user in the $i$-th cell and the $j$-th
BS. In the following, we set $D_0=R$, so $P_0$ is the average SNR
at the cell edge. We assume equal transmit power at each BS, i.e. no power control is considered\footnote{Although power control can also be used to mitigate OCI and improve the system throughput \cite{GesKia07}, the emphasis in this paper is on ICIC.}.
\item $z_i$ is the complex white Gaussian noise with zero mean and unit
variance, i.e. $z_i\sim\mathcal{CN}(0,1)$. For a general multicell network, it may include interference from distant BSs.
\item $\mathbf{h}_{i,j}$ is the $N_t\times{1}$ channel vector from the $j$-th BS
to the $i$-th user. We assume uncorrelated Rayleigh fading,
so each component of $\mathbf{h}_{i,j}$ is i.i.d.
$\mathcal{CN}(0,1)$.
\item $\mathbf{f}_{i,s_i}$ is the precoding vector for the
$i$-th user when BS $i$ takes strategy $s_i$, $i=1,2$. It is normalized, i.e. $\|\mathbf{f}_{i,s_i}\|^2=1$, and its design will be discussed later in this section.
\item $x_i$ is the transmit signal for the $i$-th user, with
the power constraint $\mathbb{E}[|x_i|^2]=1$ for $i=1,2$.
\end{itemize}

The first term on the right hand side of \eqref{eq:1stuser} and \eqref{eq:2nduser} is the information signal, while the second term is the OCI. Taking user 1 as an example, the received signal-to-interference-plus-noise ratio (SINR) is
\begin{equation}\label{eq:SINR2}
\mbox{SINR}_1(s_1,s_2)=\frac{P^r_{1,1}|\mathbf{h}_{1,1}^*\mathbf{f}_{1,s_1}|^2}{1+P^r_{1,2}|\mathbf{h}_{1,2}^*\mathbf{f}_{2,s_2}|^2}.
\end{equation}
The achievable ergodic rate is
\begin{equation}\label{eq:Rate}
R_1(s_1,s_2)=\mathbb{E}\left[\log_2\left(1+\mbox{SINR}_1(s_1,s_2)\right)\right],
\end{equation}
where $\mathbb{E}[\cdot]$ is the expectation operator.

The objective of our design is to select the strategy $s_i$ for each BS
to maximize the sum throughput, i.e. to solve the following problem
\begin{equation}\label{eq:s_select}
(s^*_1,s^*_2)=\arg\max_{s_1\in\mathcal{S}_1,s_2\in\mathcal{S}_2}{
R_1(s_1,s_2)+R_2(s_1,s_2)}.
\end{equation}
This is called \emph{adaptive ICIC}, and from \eqref{eq:SINR2} and \eqref{eq:Rate} the adaptation is based on the locations of active users, which determine $P^r_{i,j}$, $i,j=1,2$. Therefore, BSs need to exchange user locations, but instantaneous CSI of a neighboring user is needed only when ICIC is applied to suppress OCI for this user. To solve the problem in \eqref{eq:s_select} we need to first calculate the achievable sum throughput for different $(s_1,s_2)$, which will be provided in Section \ref{Sec:2cell}.
\begin{remark}
Although we use the sum throughput as the performance metric, our analysis can be easily extended to maximize a weighted sum throughput. In addition, in the following analysis and simulation, we will show that somewhat atypically, maximizing the sum throughput inherently provides fairness, and the proposed adaptive coordination strategy increases both the sum throughput and the edge throughput.
\end{remark}

\subsection{Transmission Strategies}
In this subsection, we describe the precoder design for different transmission strategies.
\subsubsection{Eigen-beamforming}
In the single-cell scenario, eigen-beamforming is optimal for the MISO system with multiple transmit and a single receive antenna \cite{Telatar99}, for which the precoding vector is the channel direction, i.e. for the $i$-th user $\mathbf{f}_{i,BF}=\mathbf{h}_{i,i}/\|\mathbf{h}_{i,i}\|$.
Therefore, the signal term is distributed as $|\mathbf{f}_{i,BF}^*\mathbf{h}_{i,i}|^2\sim\chi^2_{2N_t}$, where $\chi^2_n$ denotes the chi-square random variable (RV) with $n$ degrees of freedom.

\subsubsection{ICIC through ZF precoding}
With $N_t$ antennas each BS can maximally precancel interference for up to $N_t-1$ neighboring cells with ZF precoding. Taking cell 1 as an example, to cancel its interference for users in cell 2, 3, $\cdots$, $K$, ($K\leq{N_t}$), the precoding vector $\mathbf{f}_{1,IC}$ needs to
satisfy the orthogonality condition
$\mathbf{f}_{1,IC}^*\mathbf{h}_{i,1}=0$, for $i=2,3,\cdots,K$. Meanwhile, we also want
to maximize the desired signal power
$|\mathbf{f}_{1,IC}^*\mathbf{h}_{1,1}|^2$. This corresponds to
choosing the precoding vector $\mathbf{f}_{1,IC}$ in the direction
of the projection of vector $\mathbf{h}_{1,1}$ on the nullspace of
vectors $\hat{\mathbf{H}}=[\mathbf{h}_{2,1},\mathbf{h}_{3,1},\cdots,\mathbf{h}_{K,1}]$ \cite{JinAnd09ICC}, i.e. the precoding vector is the normalized version of the following vector
\begin{equation}
\mathbf{w}_1^{(1)}=\left(\mathbf{I}-P_{\hat{\mathbf{H}}}\right)\mathbf{h}_1^{(1)},
\end{equation}
where $P_{\hat{\mathbf{H}}}$ is the projection on $\hat{\mathbf{H}}$, given as $P_{\hat{\mathbf{H}}}=\hat{\mathbf{H}}\left(\hat{\mathbf{H}}^*\hat{\mathbf{H}}\right)^{-1}\hat{\mathbf{H}}^*$. From \cite{JinAnd09ICC}, we have the distribution of the signal power as
$|\mathbf{f}_{1,IC}^*\mathbf{h}_{1,1}|^2\sim\chi^2_{2(N_t-(K-1))}$. This ICIC strategy with ZF precoding is low complex and provides closed-form analytical results. Although MMSE precoding outperforms ZF precoding at low SNR \cite{PeeHoc05Tcomm}, as we will show later that no ICIC is required when edge SNR is low. So there is negligible performance loss associated with applying ZF precoding instead of MMSE precoding.

\subsection{Signal Power and Interference Power}
As shown in \eqref{eq:SINR2} and \eqref{eq:Rate}, the achievable throughput depends on the distributions of signal and interference terms. From the precoder design, we see that the received signal term of each user is a chi-square RV, with degrees of freedom depending on the transmission strategy of its home BS. For the interference power at the $i$-th user from the $j$-th BS, for $i\neq{j}$, if $s_j=IC(\mathcal{I}_j)$ and $i\in\mathcal{I}_j$, i.e. BS $j$ does ICIC for the $i$-th user, then user $i$ does not suffer interference from BS $j$; otherwise, the $i$-th user suffers interference distributed as $|\mathbf{f}_{j,s_j}^*\mathbf{h}_{i,j}|^2\sim\chi^2_2$, which is because the design of the precoder $\mathbf{f}_{j,s_j}$ is independent of $\mathbf{h}_{i,j}$ and $|\mathbf{f}_{j,s_j}|^2=1$. Therefore, we have the following lemma on the distribution of the received signal and interference power.
\begin{lemma}\label{Lemma:Signal}
The received signal power of the $i$-th user is distributed as
\begin{equation}
|\mathbf{f}_{i,s_i}^*\mathbf{h}_{i,i}|^2\sim\left\{\begin{array}{ll}\chi^2_{2N_t}&s_i=BF\\
\chi^2_{2(N_t-m)}&s_i=IC(\mathcal{I}_i),|\mathcal{I}_i|=m,\end{array}\right.
\end{equation}
where $|\mathcal{I}|$ is the cardinality of the set $\mathcal{I}$.

The interference power of the $i$-th user from the $j$-th BS is distributed as
\begin{equation}
|\mathbf{f}_{j,s_j}^*\mathbf{h}_{i,j}|^2\left\{\begin{array}{ll}=0&s_j=IC(\mathcal{I}_j), i\in\mathcal{I}_j\\
\sim\chi^2_2&\mbox{otherwise}.\end{array}\right.
\end{equation}
\end{lemma}

\begin{remark}
From this lemma, we see that if one BS does interference cancellation for $m$ neighboring cells instead of doing selfish beamforming, the received signal power of its own user changes from a $\chi^2_{2N_t}$ RV to a $\chi^2_{2(N_t-m)}$ RV, with the number of degrees of freedom reduced by $2m$; meanwhile, for the user in the neighboring cell helped by this BS, the interference power is reduced from a $\chi_2^2$ RV to $0$. The net effect on the sum throughput, however, is not clear. This is the focus in the following sections, i.e. to characterize the achievable sum throughput when BSs take different transmission strategies.
\end{remark}

\section{Performance Analysis of a 2-cell Network}\label{Sec:2cell}
In this section, we focus on the 2-cell network depicted in Fig. \ref{fig:2cell}. We first derive the ergodic achievable throughput with different transmission strategy pairs $(s_1,s_2)$ at two BSs, which are closed-form expressions and can be used to select $(s_1,s_2)$ to maximize the sum throughput. Then we provide some insights on the transmission strategy selection.

\subsection{Auxiliary Results}
In this subsection, we provide two lemmas that will be used in the
throughput analysis.

\begin{lemma}\label{Lemma:BF}
Assuming the RV $X$ with distribution $X\sim\chi^2_{2M}$, we have
\begin{align}\label{eq:Rate_BF}
R_{BF}(\gamma,M)=\mathbb{E}_X\left[\log_2\left(1+\gamma{X}\right)\right]
=\log_2(e)e^{1/\gamma}\sum_{k=0}^{M-1}\frac{\Gamma(-k,1/\gamma)}{\gamma^k}.
\end{align}
\end{lemma}
\begin{proof}
This result is provided as eq. (40) in \cite{AloGol99Tvt}.
\end{proof}

\begin{lemma}\label{Lemma:2cell}
Denote
\begin{equation}\notag
X\triangleq\frac{\gamma_1{Z}}{1+\gamma_2{Y}},
\end{equation}
where the RVs $Z\sim\chi^2_{2M}$, $Y\sim\chi^2_2$, and $Z$ is independent
of $Y$. Then
\begin{align}\label{eq:Rate_I2}
R_I^{(2)}(\gamma_1,\gamma_2,M)=\mathbb{E}_X\left[\log_2(1+X)\right]
=\log_2(e)\sum_{i=0}^{M-1}\sum_{l=0}^i\frac{\gamma_1^{l+1-i}}{\gamma_2(i-l)!}\cdot{I_1}\left(\frac{1}{\gamma_1},\frac{\gamma_1}{\gamma_2},i,l+1\right),
\end{align}
where $I_1$ is the integral given in \eqref{eq:I1}, with a closed-form expression given in \eqref{eq:I1expression}.
\end{lemma}
\begin{proof}
See Appendix \ref{Apx:Lemma2cell}.
\end{proof}

\subsection{Throughput Analysis}
Without loss of generality, we analyze the ergodic achievable throughput of user 1. In this part, we consider perfect CSI at the BS. The main result is given in the following theorem.

\begin{theorem}\label{Thm:2cell}
The ergodic achievable throughput of user 1 in a 2-cell network with given user locations and perfect CSI is given by
\begin{equation}\label{eq:2cell}
R_1(s_1,s_2)=\left\{\begin{array}{ll}
R_I^{(2)}(P^r_{1,1},P^r_{1,2},N_t) & (s_1,s_2)=(BF,BF)\\
R_{BF}(P^r_{1,1},N_t) & (s_1,s_2)=(BF,IC)\\
R_{BF}(P^r_{1,1},N_t-1) & (s_1,s_2)=(IC,IC)\\
R_I^{(2)}(P^r_{1,1},P^r_{1,2},N_t-1) & (s_1,s_2)=(IC,BF)
\end{array}\right.
\end{equation}
where $R_{BF}$ and $R_I^{(2)}$ are given in \eqref{eq:Rate_BF} and \eqref{eq:Rate_I2}, respectively.
\end{theorem}
\begin{proof}
The results are from \emph{Lemma \ref{Lemma:BF}} for $s_2=IC$ and \emph{Lemma \ref{Lemma:2cell}} for $s_2=BF$, together with \emph{Lemma \ref{Lemma:Signal}}.
\end{proof}

The results in \emph{Theorem \ref{Thm:2cell}} are closed-form expressions, from which we are able to select the strategy pair to maximize the sum throughput. However, the expressions in \eqref{eq:Rate_BF} and \eqref{eq:Rate_I2} are complicated and provide little insight. In the following, we provide a heuristic discussion on the strategy selection
for different interference-to-noise ratio (INR) scenarios.

\textbf{Both users are noise-limited:} This scenario corresponds to INR$_1\ll{1}$ and INR$_2\ll{1}$. It may happen when both users are in the cell interior, or when the edge SNR is very low. For this scenario, as noise dominates OCI, ICIC provides a marginal gain, and each BS is willing to do beamforming to increase the received signal power for its own user, i.e. the strategy pair will be $(BF,BF)$. Therefore, \emph{there is no need to do ICIC in this scenario}.

\textbf{Both users are interference-limited:} This scenario corresponds to INR$_1\gg{1}$ and INR$_2\gg{1}$. This may happen when both users are at the cell edge and the transmit power is relatively high compared to the additive noise. As users suffer a higher level of OCI in this scenario, the BS will do ICIC for the neighboring cell to increase the sum throughput, i.e. the strategy pair will be $(IC,IC)$.

\textbf{One user is noise-limited, and the other is interference-limited:} This scenario corresponds to INR$_1\ll{1}$ and INR$_2\gg{1}$. This may happen when user 1 is in the cell interior and user 2 is at the cell edge. For the interior user, it normally enjoys a high SINR, so its throughput is limited by bandwidth. This means that doing ICIC for user 2 will not hurt user 1 so much, as the received signal power reduction for user 1 only brings a throughput loss in a $\log$ scale. On the other hand, user 2 is limited by OCI, so it requires ICIC from BS 1. Meanwhile, BS 2 will do beamforming for user 2 to increase the signal power, as the throughput of user 2 is power-limited. Therefore, the strategy pair will be $(IC,BF)$.

\begin{remark}
Although this is just a heuristic discussion, it shows that different strategy pairs will be selected for different scenarios, depending on user locations and average edge SNR. The ICIC strategy is not always necessary. The third scenario is of particular interest, as it shows that even with sum throughput as the metric the BS with an interior user (high rate) is willing to help the edge user (low rate) in the neighboring cell, i.e. \emph{encouraging fairness}. Note that the strategy pair selection in the above discussion may not be the actual selection, and the actual strategy depends on user locations, the additive noise level and edge SNR, which can be determined from \eqref{eq:2cell}.
\end{remark}

In Fig. \ref{fig:simvscal_10dB}, we compare the simulation and
calculation results. Referring to Fig. \ref{fig:2cell}, user 1 is fixed at the cell edge $(-0.1R,0)$,
while user 2 is moving on the line connecting BS 1 and BS 2, with
location $(x_2R,0)$. We see that for average edge SNR $P_0=10$ dB, and for the
considered locations, the strategy pair $(IC,IC)$ is always
selected. In Fig. \ref{fig:CSIT}, we plot the selected strategy pairs for
different user locations, where user 1 and 2 are moving on the line
connecting BS 1 and BS 2. The x- and y-axis are the distance for
user 1 and user 2 from the central point $(0,0)$, respectively. The following
observations can be made:
\begin{enumerate}
\item When the edge SNR is small ($P_0=-5$ dB), $(BF,BF)$ dominates, as the throughput is limited by noise and each BS tries to increase the received signal power for its own user.
\item When the edge SNR is large ($P_0=10$ dB), $(IC,IC)$ dominates, as the throughput is limited by OCI and each BS does ICIC for neighboring cells.
\item For medium SNR ($P_0=5$ dB), the selected strategy pair depends on the user locations. Specifically, it shows that when both users are in cell interior, i.e. INRs are small, $(BF,BF)$ is selected; when both users are at cell edge, i.e. INRs are large, $(IC,IC)$ is selected; when one user is in cell interior, and the other is at cell edge, the BS with the interior user will do ICIC for the edge user.
\end{enumerate}
These observations agree with the above discussion and motivate to adaptively select transmission strategies.

\section{From 3-cell to Multicell Networks}\label{Sec:3cell}
The investigation of the simplified 2-cell network provided insights about the strategy selection and motivated the adaptive coordination, but the result cannot be readily implemented in a general multicell network. In this section, we first extend our adaptive strategy to a 3-cell network. We derive closed-form expressions for the achievable throughput for the strategy selection. Based on the results for 3-cell networks, we also propose approaches to extend the adaptive coordination to general multicell networks.

\subsection{The Strategy Set}
With 3 cells coordinating with each other, each BS has four
different strategies. Taking user 1 as an example, we describe
different strategies as follows.
\begin{enumerate}
\item \textbf{Selfish beamforming:} BS 1 does beamforming for user 1, denoted as $s_1=BF$.
\item \textbf{ICIC for 2 neighboring cells:} BS 1 does ICIC for both
cell 2 and 3, which requires $N_t\geq{3}$. This is denoted as
$s_1=IC(\{2,3\})$.
\item \textbf{ICIC for cell 2:} BS 1 does ICIC for cell 2, denoted as $s_1=IC(2)$.
\item \textbf{ICIC for cell 3:} BS 1 does ICIC for cell 3, denoted as $s_1=IC(3)$.
\end{enumerate}
To reduce the size of the strategy set, we combine strategy 3 and 4
as a single strategy, for which BS 1 does ICIC for the neighboring
cell that suffers a higher level of average OCI from BS 1, i.e. to
help the neighboring cell user that is closer to BS 1. This is a reasonable approach and reduces the complexity of the strategy selection process.

Therefore, the strategy set for user 1 is $\bar{\mathcal{S}}_1=\{BF, IC(2\mbox{ or }3), IC(\{2,3\})\}$. There are a total of $3^3=27$ different strategy combinations for 3 users, $(s_1,s_2,s_3)\in\bar{\mathcal{S}}_1\times\bar{\mathcal{S}}_2\times\bar{\mathcal{S}}_3$.

\subsection{Throughput Analysis}
First, we present the following lemma for throughput analysis.
\begin{lemma}\label{Lemma:3cell}
Denote
\begin{equation}
X\triangleq\frac{\alpha{Z}}{1+\delta_1{Y_1}+\delta_2{Y_2}},
\end{equation}
where $Z\sim\chi^2_{2M}$, $Y_1\sim\chi^2_{2}$,
$Y_2\sim\chi^2_{2}$, and they are independent. Then
\begin{align}\label{eq:Rate_3cell}
&R_I^{(3)}(\alpha,\delta_1,\delta_2,M)=\mathbb{E}_X\left[\log_2(1+X)\right]\notag\\
=&\log_2(e)\sum_{i=0}^{M-1}\sum_{l=0}^i\frac{\alpha^{l-i+1}}{(\delta_1-\delta_2)(i-l)!}\left[I_1\left(\frac{1}{\alpha},\frac{\alpha}{\delta_1},i,l+1\right)-I_1\left(\frac{1}{\alpha},\frac{\alpha}{\delta_2},i,l+1\right)\right],
\end{align}
where $I_1(\cdot,\cdot,\cdot,\cdot)$ is the integral given in \eqref{eq:I1}, with a closed-form expression given in \eqref{eq:I1expression}.
\end{lemma}
\begin{proof}
The proof is similar to the one in Appendix \ref{Apx:Lemma2cell} for \emph{Lemma \ref{Lemma:2cell}}.
\end{proof}

Taking the first user as an example, its received SINR with the strategy $\mathbf{s}=(s_1,s_2,s_3)$ is
\begin{equation}\label{eq:SINR_TTT}
\mbox{SINR}_1(\mathbf{s})=\frac{P^r_{1,1}|\mathbf{h}_{1,1}^*\mathbf{f}_{1,s_1}|^2}{1+P^r_{1,2}|\mathbf{h}_{1,2}^*\mathbf{f}_{2,s_2}|^2+P^r_{1,3}|\mathbf{h}_{1,3}^*\mathbf{f}_{3,s_3}|^2}.
\end{equation}
The achievable rate is
\begin{equation}\label{R_3cell}
R_1(\mathbf{s})=\mathbb{E}\left[\log_2\left(1+\mbox{SINR}_1(s_1,s_2,s_3)\right)\right],
\end{equation}
for which a closed-form expression is given in the following theorem.
\begin{theorem}\label{Thm:3cell}
The ergodic achievable throughput of user 1 in a 3-cell network with given user locations and perfect CSI is given by
\begin{equation}\label{eq:3cell}
R_1(\mathbf{s})=\left\{\begin{array}{ll}
R_I^{(3)}(P^r_{1,1},P^r_{1,2},P^r_{1,3},M) & s_2=BF, s_3=BF\\
R_I^{(2)}(P^r_{1,1},P^r_{1,j},M) & s_j=IC(\mathcal{I}_j),1\in\mathcal{I}_j, j=2\mbox{ or }3\\
R_{BF}(P^r_{1,1},M) & s_j=IC(\mathcal{I}_j),1\in\mathcal{I}_j, j=2,3
\end{array}\right.
\end{equation}
where $R_{BF}$, $R_I^{(2)}$, and $R_I^{(3)}$ are given in \eqref{eq:Rate_BF}, \eqref{eq:Rate_I2}, and \eqref{eq:Rate_3cell} respectively. The parameter $M$ depends on the distribution of the signal term, which subsequently depends on $s_1$:
\begin{equation}\label{eq:M}
M=\left\{\begin{array}{ll}N_t & s_1=BF \\N_t-m & s_1=IC(\mathcal{I}_1),|\mathcal{I}_1|=m.\end{array}\right.
\end{equation}
\end{theorem}
\begin{proof}
The results come from \emph{Lemma \ref{Lemma:Signal}} and \emph{Lemma \ref{Lemma:3cell}}.
\end{proof}

Based on this theorem, we are able to select the transmission strategy at each BS to maximize the sum throughput. Note that the strategy selection is in a coordinated way, i.e. the 3 BSs jointly determine the set $(s_1,s_2,s_3)$, as the objective function is common for all the BSs. It explicitly assumes that each BS knows the strategy taken by other BSs.

\subsection{Extension to Multicell Networks}\label{Sec:Distri}
In this subsection, we propose approaches to extend our results to a general multicell setting. A detailed investigation is beyond the scope of this paper but is feasible in principle.

One approach is to apply the proposed adaptive ICIC strategy with cell sectoring, as in \cite{Zhang04}. By using 120-degree sectoring in each cell, every 3 neighboring cells can coordinate with each other to serve users in the shadow area shown in Fig. \ref{fig:3cell}, where the 3 BSs jointly select the transmission strategy based on the results in \emph{Theorem \ref{Thm:3cell}}.

It is also possible to implement the adaptive ICIC strategy in a distributed way. In this approach, each BS determines its transmission strategy independently rather than in a coordinated way. The main idea is for each BS to select its transmission strategy by itself. To do this, each BS needs to estimate if there is a sum throughput gain by providing ICIC for its neighboring cells. If the sum throughput is increased, it will select ICIC as its strategy; otherwise, it will perform beamforming for its own user. As each user is located in the interior area of a certain 3-cell sub-network, as in the shadow area in Fig. \ref{fig:3cell}, its achievable throughput can be estimated based on \eqref{eq:3cell} by approximating the interference from outer cells as white Gaussian noise. No cluster structure is used, so this approach can be adopted in a network of an arbitrary size.

\section{Impact of Limited Feedback}\label{Sec:LFB}
We have assumed perfect CSI at the BS in the results provided thus far. However, in realistic scenarios, there will always be inaccuracy in the available CSI. In this section, we consider a FDD (Frequency Division Duplex) system where CSI is obtained through limited feedback \cite{LovHea08JSAC}. As limited feedback sends quantized channel information to the transmitter, it introduces quantization error to the available CSI. We will analyze the impact of limited feedback, and consider feedback design for adaptive ICIC transmission.

\subsection{Limited Feedback}
With limited feedback, the channel direction information (CDI) is fed back using a quantization
codebook known at both the transmitter and receiver. The
quantization is chosen from a codebook of unit norm vectors of size
$L=2^B$, where $B$ is the number of feedback bits. Denote the codebook as
$\mathcal{C}=\{\mathbf{c}_{1},\mathbf{c}_{2},\cdots,\mathbf{c}_{L}\}$.
Each user quantizes its channel direction to the closest codeword,
measured by the inner product. Therefore, the quantized channel
direction is
\begin{equation}
\hat{\mathbf{h}}_{i,j}=\arg\max_{\mathbf{c}\in\mathcal{C}}|\tilde{\mathbf{h}}_{i,j}^*\mathbf{c}|,
\end{equation}
where
$\tilde{\mathbf{h}}_{i,j}=\frac{\mathbf{h}_{i,j}}{\|\mathbf{h}_{i,j}\|}$
is the actual channel direction. Then each user feeds back $B$ bits to indicate the index of this codeword in the codebook $\mathcal{C}$. We assume the channel estimation at each user is perfect, and the feedback channel is error-free and without delay. Random vector quantization (RVQ) \cite{SanHon04ISIT,Jin06IT} is used to facilitate the analysis, where each quantization vector is independently chosen from the isotropic distribution on the $N_t$-dimensional unit sphere.

If ICIC is performed for the $i$-th user by some of its neighboring BSs, this user needs to estimate channel directions from multiple BSs, which are then independently quantized and fed back to its home BS. Then the home BS can forward the associated CDI to neighboring BSs through backhaul connection.
\begin{AS}
The $i$-th user uses the codebook $\mathcal{C}_{i,j}$ to quantize CDI for the $j$-th BS, which is of size $L_{i,j}=2^{B_{i,j}}$. If $L_{i,j}$ is the same for different $j$, user $i$ can use the same quantization codebook, but the codebooks are different from user to user.
\end{AS}

As will be shown later, the quantization for channel directions of different BSs have different impacts on the system performance, so different $L_{i,j}$ for different $i$ and $j$ may provide better performance. Different users employing different codebooks is to avoid the same quantized CDI from multiple users at the same BS.

\subsection{Throughput Analysis}
First, we consider the statistics of the quantized CDI. Let
$\cos\theta_{i,j}=|\tilde{\mathbf{h}}_{i,j}^*\hat{\mathbf{h}}_{i,j}|$, where $\theta_{i,j}=\angle\left(\tilde{\mathbf{h}}_{i,j},\hat{\mathbf{h}}_{i,j}\right)$,
then we have \cite{YueLov07Twc}
\begin{equation}\label{eq:xi}
\xi_{i,j}=\mathbb{E}_{\theta_{i,j}}\left[\cos^2\theta_{i,j}\right]=1-L_{i,j}\cdot\beta\left(L_{i,j},\frac{N_t}{N_t-1}\right),
\end{equation}
where $\beta(x,y)$ is the Beta function, i.e.
$\beta(x,y)=\frac{\Gamma(x)\Gamma(y)}{\gamma(x+y)}$ with
$\Gamma(x)=\int_0^\infty{t}^{x-1}e^{-t}\mbox{d}t$ as the Gamma function.

To investigate the impact of limited feedback, we first analyze the received signal power and interference power with limited feedback.

\begin{lemma}\label{Lemma:Signal_LFB}
If CDI at the BS is obtained through limited feedback, the received signal power of the $i$-th user with the expectation on $\theta_{i,i}$ can be approximated as $\mathbb{E}_{\theta_{i,i}}\left[|\mathbf{h}^*_{i,i}\mathbf{f}_{i,s_i}|^2\right]\approx\xi_{i,i}{X}$, where $\xi_{i,i}$ is given in \eqref{eq:xi} and the RV $X$ is distributed as
\begin{equation}
X\sim\left\{\begin{array}{ll}\chi^2_{2N_t}&s_i=BF\\
\chi^2_{2(N_t-m)}&s_i=IC(\mathcal{I}_i),|\mathcal{I}_i|=m.\end{array}\right.
\end{equation}
The interference power of the $i$-th user from the $j$-th BS is distributed as
\begin{equation}
|\mathbf{h}^*_{i,j}\mathbf{f}_{j,s_j}|^2\sim\left\{\begin{array}{ll}\kappa_{i,j}\chi^2_2&s_j=IC(\mathcal{I}_j), i\in\mathcal{I}_j\\
\chi^2_2&\mbox{otherwise},\end{array}\right.
\end{equation}
where $\kappa_{i,j}=2^{-\frac{B_{i,j}}{N_t-1}}$.
\end{lemma}
\begin{proof}
See Appendix \ref{Apx:Lemma_SignalLFB}.
\end{proof}
\begin{remark}
From this lemma, we see that limited feedback has differing impact on the received signal term and the interference term: it only changes the mean, not the distribution of the signal term; for the interference term, the distribution is the same without ICIC, but limited feedback causes residual interference with ICIC. In addition, at high edge SNR, the impact of limited feedback on the signal term only causes a constant rate loss of $\log\xi_{i,i}$ for the $i$-th user, but the resulting residual OCI increases with edge SNR and limits the system throughput. Therefore, the CDI need not be of the same accuracy for the home BS and the helper BS, which leaves flexibility for the feedback design.
\end{remark}

Based on the above lemma, we provide the following theorem on the achievable throughput with limited feedback.
\begin{theorem}\label{Thm:2cell_LFB}

The achievable throughput of user 1 in a 3-cell network with given user locations and limited feedback is approximated by
\begin{equation}\label{eq:3cell_LFB}
R_1(\mathbf{s})\approx\left\{\begin{array}{ll}
R_I^{(3)}(\xi_{1,1} P^r_{1,1},P^r_{1,2},P^r_{1,3},M) & s_2=BF, s_3=BF\\
R_I^{(3)}(\xi_{1,1} P^r_{1,1},P^r_{1,j}, \kappa_{1,k}P^r_{1,k}, M) & s_j=IC(\mathcal{I}_j),1\in\mathcal{I}_j, j=2\mbox{ or }3\\
R_I^{(3)}(\xi_{1,1} P^r_{1,1},\kappa_{1,2}P^r_{1,2},\kappa_{1,3}P^r_{1,3}M) & s_j=IC(\mathcal{I}_j),1\in\mathcal{I}_j, j=2,3
\end{array}\right.
\end{equation}
where $R_I^{(3)}$ is given in \eqref{eq:Rate_3cell}, and $M$ is given by \eqref{eq:M}.
\end{theorem}
\begin{proof}
See Appendix \ref{Apx:Thm2cell_LFB}.
\end{proof}
\begin{remark}
This result can be easily modified for a 2-cell network. Note that with limited feedback, each user always suffers from OCI, due to co-channel transmission and/or imperfect interference cancellation. Therefore, the transmission strategy selection now depends not only on user locations, average edge SNR, but also on the number of feedback bits.
\end{remark}

In Fig. \ref{fig:SimvsApprox_B10dB10}, we compare the simulation and approximation results with limited feedback in the same setting as Fig. \ref{fig:simvscal_10dB} and with feedback bits for each channel direction to be $B=10$. We see that the approximations are very accurate. Compared to Fig. \ref{fig:simvscal_10dB}, the performance gain due to ICIC is reduced, but the strategy pair $(IC,IC)$ is still preferred. We can also get a similar plot as Fig. \ref{fig:CSIT}, which is omitted due to space limitation, but similar observations can be made except that operating regions with the ICIC strategy shrink.

\subsection{Limited Feedback Design}
With ICIC, each user needs to feed back multiple channel directions. The feedback should be carefully designed as the resource on the feedback channel is limited. In this subsection, we consider feedback in the following two scenarios:
\begin{itemize}
\item If the number of feedback bits can be varying, how many bits do we need to keep a constant rate loss versus the perfect CSI case?
\item If the total number of feedback bits is fixed, how should we allocate them between the CSI feedback for the home BS and the CSI feedback for the helper BS?
\end{itemize}

\subsubsection{Feedback bits for a constant rate loss}
If we can vary the number of feedback bits, based on the rate loss analysis, we provide the following theorem on the required scaling of feedback bits with different system parameters to keep a constant rate loss.
\begin{theorem}\label{Thm:scaleB}
In a 3-cell network, to keep a constant rate loss of $\log_2\delta_R$ bps/Hz compared to perfect CSI, the number of feedback bits for each helper BS needs to satisfy
\begin{equation}\label{eq:scaleB}
B^\star\geq(N_t-1)\log_2\left(\frac{2P_0}{\delta_R-1}\right).
\end{equation}
\end{theorem}
\begin{proof}
See Appendix \ref{Apx:ThmScaleB}.
\end{proof}

We see that similar to multiuser MIMO systems \cite{Jin06IT,ZhaRob09EURASIP}, the feedback bit rate needs to increase linearly with both $N_t$ and $P_0$ (in dB). The difference is that for ICIC there is no such requirement on the feedback of the CSI to the home BS, as it provides a fixed rate loss with a fixed number of feedback bits at high SNR according to \emph{Lemma \ref{Lemma:Signal_LFB}}, so we do not to increase the feedback bits for this link.

\subsubsection{Feedback bits allocation}
As shown in \emph{Lemma \ref{Lemma:Signal_LFB}}, the CSI accuracy for the home BS and the helper BS has different impact on the performance. This indicates that with a fixed number of feedback bits it is possible to improve the performance by adaptively allocating the total feedback bits for the home BS and the helper BS rather than an equal allocation.

With feedback bits allocation, the maximum achievable throughput for a given scenario is now given as
\begin{equation}\label{eq:bitAlloc}
R_{\mbox{sum}}=\max_{\mathbf{s}\in\mathcal{S},\sum_jB_{i,j}=B}\sum_{i=1}^3R_i(\mathbf{s},B_{i,1},B_{i,2},B_{i,3}),
\end{equation}
where $\mathbf{s}=(s_1,s_2,s_3)$ and the expression for $R_i$ is given in \eqref{eq:3cell_LFB}. This is a combinatorial optimization problem, but may be solved by an exhaustive search for small $B$. To reduce the search space, we can add additional constraints such as forcing the number of feedback bits to be the same for all the helper BSs.

\section{Numerical Results}\label{Sec:Num}
In this section, we present some numerical results to show the performance of our proposed adaptive ICIC strategy and build intuition. A 3-cell network as shown in Fig. \ref{fig:3cell} is considered, where there is one active user randomly located in each cell in the shadow area. The radius of each cell is $R=1$ km, the path loss exponent is $3.7$, and $N_t=4$.

\subsection{Performance comparison with perfect CSI}
In this part, we compare the performance of three systems with different transmission strategies: the system without ICIC, i.e. each BS does selfish beamforming for its own user, denoted as \emph{no ICIC}, the system where each BS always does ICIC for both of its neighboring cells, denoted as \emph{static ICIC}, and the system with the proposed adaptive ICIC strategy where the transmission strategy at each BS is selected based on \emph{Theorem \ref{Thm:3cell}}, denoted as \emph{adaptive ICIC}.

Fig. \ref{fig:Comp} shows the average throughput and the 5th percentile throughput, representing the cell edge throughput, for each of the three systems. We see that the performance difference depends on the edge SNR $P_0$:

\textbf{For low $P_0$}, the \emph{static ICIC} system provides lower throughput than the other two, which shows there is no need to perform ICIC in this regime.

\textbf{For high $P_0$}, there is a rate ceiling for the \emph{no ICIC} system, as its throughput is limited by OCI. The \emph{static ICIC} and \emph{adaptive ICIC} systems have similar performance, and both provide significant throughput gain over the \emph{no ICIC} system, e.g for $P_0=15$ dB\footnote{As shown in \cite{HuaVal05PIMRC}, an edge SNR$\sim15$ dB can be obtained with reasonable assumptions.}, the average throughput gain is $53\%$ while the edge throughput gain is $210\%$.

\textbf{For medium $P_0$ ($0\sim5$)}, \emph{adaptive ICIC} outperforms both \emph{no ICIC} and \emph{static ICIC}, which shows that we should not simply switch between selfish beamforming and static ICIC but should adaptively and jointly select the transmission technique at each BS.

Compared to \emph{static ICIC}, the \emph{adaptive ICIC} system is able to reduce the amount of required CSI, as the CSI for a neighboring BS is needed only when this BS does ICIC for the user. Fig. \ref{fig:FB} compares the amount of CSI requirement for different systems, in number of channel directions. For medium $P_0$, \emph{adaptive ICIC} reduces the required CSI amount compared to \emph{static ICIC} while provides higher throughput than the other two systems.

\subsection{Impact of limited feedback}
In Fig. \ref{fig:Comp}, we show the system performance when CSI is provided by limited feedback. We assume each user feeds back $B_s$ bits for the CSI to its home BS and $B_I$ bits for the CSI to each of its helper BSs if required. Adaptive ICIC based on \eqref{eq:3cell_LFB} is applied.

First, we scale $B_I$ according to \eqref{eq:scaleB} with the rate loss target $\delta_R=1$ bps/Hz, i.e. $B_I=B^\star$, while $B_s$ has the same scaling ($B_s=B^\star$) or is fixed to be $B_s=6$. It shows that by scaling $B_I$, we can obtain a constant rate loss (less than $1$ bps/Hz for both average and edge throughput) to the system with perfect CSI. In addition, we do not need to scale $B_s$, as a fixed $B_s$ only causes a fixed rate loss at high edge SNR. Note that at high $P_0$, a large value of $B_I$ is required, e.g. $B_I=18$ for $P_0=15$.

Next, we assume the total number of feedback bits is fixed to be $30$ bits, i.e. $B_s+2B_I=30$, and compare adaptive bit allocation according to \eqref{eq:bitAlloc} with uniform allocation ($B_s=B_I=10$). For adaptive bit allocation, each user and BS need to maintain multiple codebooks, and to reduce the complexity we limit the allocation bit pair as $(B_s,B_I)\in\{(10,10),(8,11),(6,12),(4,13),(2,14)\}$. We see that adaptive bit allocation provides better performance at high $P_0$, but the throughput gain over uniform allocation is marginal ($\sim6\%$ for average throughput and $\sim11\%$ for edge throughput). This is because the total number of feedback bits is not large enough. For example, to keep a rate loss of $1$ bps/Hz to perfect CSI case, we need $B_I=18$ for $P_0=15$, so $B_s+2B_I>36$, which can not be satisfied with the available number of bits. However, we see that all the adaptive ICIC systems provide a significant gain for both average and edge throughput over the system without ICIC even with limited feedback.

\section{Conclusions}\label{Sec:Con}
In this paper, we investigated spatial ICIC to suppress OCI in a multicell wireless network. An adaptive strategy was proposed, where multiple BSs jointly select transmission strategies. ICIC is a type of coordinated single-cell transmission, so the system complexity is low and no central processing unit is required. In addition, it has a low overhead as only user locations are required for strategy selection, and instantaneous CSI is needed only at the home BS and the neighboring BS that does ICIC for the user. Numerical results showed that ICIC provides both average and edge throughput gain, and at medium edge SNR the adaptive strategy outperforms both the system without ICIC and the one with static ICIC. Even when the CSI is provided by limited feedback, the ICIC system still provides significant throughput gain with carefully designed feedback strategies. Given the consistent results for two and three cells, we conjecture that the results and design intuition extend to large cellular networks.

\useRomanappendicesfalse
\appendix
\subsection{Proof of Lemma \ref{Lemma:2cell}}\label{Apx:Lemma2cell}
Due to space limitation, we only provide some key steps. First, the cumulative distribution function (cdf) of the
RV $X$ can be derived as
\begin{align}
F_X(x)=1-\sum_{i=0}^{M-1}\sum_{l=0}^i\frac{\gamma_1^{l+1-i}}{\gamma_2(i-l)!}\cdot\frac{x^ie^{-x/\gamma_1}}{\left(x+\frac{\gamma_1}{\gamma_2}\right)^{l+1}}.
\end{align}
The expectation of $\ln(1+x)$ on $X$ is then derived as follows.
\begin{align}
&\mathbb{E}_X\left[\ln(1+X)\right]=\int_0^\infty\ln(1+x)\mbox{d}F_X
\stackrel{(a)}{=}\int_0^\infty\frac{1-F_X(x)}{x+1}\mbox{d}x\notag\\
=&\sum_{i=0}^{M-1}\sum_{l=0}^i\frac{\gamma_1^{l+1-i}}{\gamma_2(i-l)!}\int_0^\infty\frac{x^ie^{-x/\gamma_1}}{(x+1)\left(x+\frac{\gamma_1}{\gamma_2}\right)^{l+1}}\mbox{d}x
=\sum_{i=0}^{M-1}\sum_{l=0}^i\frac{\gamma_1^{l+1-i}}{\gamma_2(i-l)!}\cdot{I_1}\left(\frac{1}{\gamma_1},\frac{\gamma_1}{\gamma_2},i,l+1\right),
\end{align}
where step (a) follows integration by parts and
$I_1(\cdot,\cdot,\cdot,\cdot)$ is the integral
\begin{equation}\label{eq:I1}
I_1(a,b,m,n)=\int_0^\infty\frac{x^me^{-ax}}{(x+b)^n(x+1)}\mbox{d}x.
\end{equation}
Then we get \eqref{eq:Rate_I2}. A closed-form expression for the integral $I_1$ is given as follows:
\begin{equation}\label{eq:I1expression}
{I}_1(a,b,m,n)=\sum_{i=1}^n\frac{(-1)^{i-1}}{(1-b)^i}\cdot{{I}_2}\left(a,b,m,n-i+1\right)+\frac{{{I}_2}\left(a,1,m,1\right)}{(b-1)^n},
\end{equation}
where $I_2$ is given as
\begin{equation}
{I}_2(a,b,m,n)=\int_0^\infty\frac{x^me^{-ax}}{(x+b)^n}\mbox{d}x
=e^{ab}\sum_{i=0}^m{m\choose{i}}(-b)^{m-i}{I}_3(a,b,i-n),
\end{equation}
and
\begin{equation}
{I}_3(a,b,m)=\int_b^\infty{x^me^{-ax}}\mbox{d}x
=\left\{\begin{array}{lc}e^{-ab}\sum_{i=0}^m\frac{m!}{i!}\frac{b^i}{a^{m-i+1}}&m\geq{0}\\
E_1(ab) & m=-1\\
\frac{E_1(ab)}{(-a)^{-m-1}(-m-1)!}+\frac{e^{-ab}}{b^{-m-1}}\sum_{i=0}^{-m-2}\frac{(-ab)^i(-m-i-2)!}{(-m-1)!} & m\leq{-2}\end{array},\right.
\end{equation}
where $E_1(x)$ is the exponential-integral function of the first
order.

\subsection{Proof of Lemma \ref{Lemma:Signal_LFB}}\label{Apx:Lemma_SignalLFB}
For the signal term, if $s_i=BF$, the beamforming vector is
now based on the quantized channel direction, i.e. $\mathbf{f}_{i,BF}=\hat{\mathbf{h}}_{i,i}$. The signal term is
\begin{equation}
|\mathbf{h}_{i,i}^*\mathbf{f}_{i,BF}|^2=|\mathbf{h}_{i,i}^*\hat{\mathbf{h}}_{i,i}|^2=\cos^2\theta_{i,i}\|\mathbf{h}_{i,i}\|^2.
\end{equation}
Taking the expectation on $\theta_{i,i}$, we have
$\mathbb{E}_{\theta_{i,i}}\left[|\mathbf{h}_{i,1}^*\mathbf{f}_{i,BF}|^2\right]\sim\xi_{i,i}\chi^2_{2N_t}$.

If $s_i=IC(\mathcal{I}_i)$, $|\mathcal{I}_i|=m$, writing
$\tilde{\mathbf{h}}_{i,i}=(\cos\theta_{i,i})\hat{\mathbf{h}}_{i,i}+(\sin\theta_{i,i})\mathbf{g}_{i,i}$,
where $\mathbf{g}_{i,i}$ is orthogonal to
$\hat{\mathbf{h}}_{i,i}$. The signal term with an expectation on
$\theta_{i,i}$ is
\begin{align}\label{eq:Signal_IC}
\mathbb{E}_{\theta_{i,i}}\left[|\mathbf{h}_{i,i}^*\mathbf{f}_{i,IC}|^2\right]&=\|\mathbf{h}_{i,i}\|^2\cdot\mathbb{E}_{\theta_{i,i}}|(\cos\theta_{i,i})\hat{\mathbf{h}}_{i,i}^*\mathbf{f}_{i,IC}+(\sin\theta_{i,i})\mathbf{g}_{i,i}^*\mathbf{f}_{i,IC}|^2\notag\\
&\stackrel{(a)}{\approx}\|\mathbf{h}_{i,i}\|^2\cdot\mathbb{E}_{\theta_{i,i}}|(\cos\theta_{i,i})\hat{\mathbf{h}}_{i,i}^*\mathbf{f}_{i,IC}|^2\notag\\
&=\mathbb{E}_{\theta_{i,i}}\left[\cos^2\theta_{i,i}\right]\cdot\|\mathbf{h}_{i,i}\|^2|\hat{\mathbf{h}}_{i,i}^*\mathbf{f}_{i,IC}|^2\notag\\
&\stackrel{(b)}{\sim}\xi_{i,i}\chi^2_{2(N_t-m)}.
\end{align}
In step (a), we remove the $\sin\theta_{i,i}$ term, which is
normally very small. The
beamforming vector $\mathbf{f}_{i,IC}$ is in the direction of the
projection of vector $\hat{\mathbf{h}}_{i,i}$ on the nullspace of
$\hat{\mathbf{h}}_{j,i}$, $\forall{j}\neq{i}$, so similar to the perfect CSI case we
have
$\|\mathbf{h}_{i,i}\|^2|\hat{\mathbf{h}}_{i,i}^*\mathbf{f}_{i,IC}|^2\sim\chi^2_{2(N_t-m)}$,
which gives step (b).

For the interference power, take user 1 for an example, and consider the interference from BS 2. If BS 2 does not apply ICIC for user 1, then as in \emph{Lemma \ref{Lemma:Signal}}, $|\mathbf{h}^*_{1,2}\mathbf{f}_{2,IC}|^2\sim\chi_2^2$. If BS 2 uses ICIC for user 1, with quantization error, there will be residual interference from BS 2. The interference power is
$|\mathbf{h}^*_{1,2}\mathbf{f}_{2,IC}|^2$, where
$\mathbf{f}_{2,IC}$ is in the direction of the projection of vector
$\hat{\mathbf{h}}_{2,2}$ on the nullspace of
$\hat{\mathbf{h}}_{1,2}$. Based on the quantization cell
approximation, this interference term can be approximated as an
exponential RV with mean $\kappa_{i,j}=2^{-\frac{B_{i,j}}{N_t-1}}$
\cite{ZhaRob09EURASIP}, i.e. $|\mathbf{h}^*_{1,2}\mathbf{f}_{2,IC}|^2\sim\kappa_{i,j}\chi^2_2$.

\subsection{Proof of Theorem \ref{Thm:2cell_LFB}}\label{Apx:Thm2cell_LFB}
In a 3-cell network, the achievable rate for user 1 is first approximated as
\begin{align}
R_1(\mathbf{s})&=\mathbb{E}_{\mathbf{h},\theta_{1,1}}\left[\log_2\left(1+\frac{P^r_{1,1}|\mathbf{h}_{1,1}^*\mathbf{f}_{1,s_1}|^2}{1+P^r_{1,2}|\mathbf{h}_{1,2}^*\mathbf{f}_{2,s_2}|^2+P^r_{1,3}|\mathbf{h}_{1,3}^*\mathbf{f}_{3,s_3}|^2}\right)\right]\notag\\
&\approx\mathbb{E}_\mathbf{h}\left[\log_2\left(1+\frac{P^r_{1,1}\mathbb{E}_{\theta_{1,1}}\left[|\mathbf{h}_{1,1}^*\mathbf{f}_{1,s_1}|^2\right]}{1+P^r_{1,2}|\mathbf{h}_{1,2}^*\mathbf{f}_{2,s_2}|^2+P^r_{1,3}|\mathbf{h}_{1,3}^*\mathbf{f}_{3,s_3}|^2}\right)\right].
\end{align}
Then based on \emph{Lemma \ref{Lemma:3cell}} and \emph{Lemma \ref{Lemma:Signal_LFB}}, we get the results in \eqref{eq:3cell_LFB}.

\subsection{Proof of Theorem \ref{Thm:scaleB}}\label{Apx:ThmScaleB}
At high edge SNR, each BS is doing ICIC for both its neighboring cells, and the system throughput is limited by the residual OCI. As the two neighboring cells are symmetric, let the number of feedback bits be $B_I$ for each of them. As shown in \cite{Jin06IT,ZhaRob09EURASIP}, the rate loss due to imperfect CSI is upper bounded as
\begin{align}
\Delta{R}&\leq\mathbb{E}\left[\log_2\left(1+P^r_{1,2}|\mathbf{h}_{1,2}^*\mathbf{f}_{2,IC}|^2+P^r_{1,3}|\mathbf{h}_{1,3}^*\mathbf{f}_{3,IC}|^2\right)\right]\notag\\
&\stackrel{(a)}{\leq}\log_2\left(1+P^r_{1,2}\mathbb{E}\left[|\mathbf{h}_{1,2}^*\mathbf{f}_{2,IC}|^2\right]+P^r_{1,3}\mathbb{E}\left[|\mathbf{h}_{1,3}^*\mathbf{f}_{3,IC}|^2\right]\right)\notag\\
&\stackrel{(b)}{=}\log_2\left(1+P^r_{1,2}2^{-\frac{B_I}{N_t-1}}+P^r_{1,3}2^{-\frac{B_I}{N_t-1}}\right)\notag\\
&\stackrel{(c)}{\leq}\log_2\left(1+2P_0\cdot2^{-\frac{B_I}{N_t-1}}\right),
\end{align}
where step (a) follows Jensen's inequality, step (b) is from \emph{Lemma \ref{Lemma:Signal_LFB}}, and step (c) is due to the fact $P^r_{1,2}\leq{P_0}$ and $P^r_{1,3}\leq{P_0}$. Then by solving
\begin{equation}
\log_2\left(1+2P_0\cdot2^{-\frac{B_I}{N_t-1}}\right)=\log_2\delta_R
\end{equation}
we get the result in \eqref{eq:scaleB}.

\bibliographystyle{IEEEtran}
\bibliography{bibi}

\begin{figure}
\centering
\includegraphics[clip=true,scale=1]{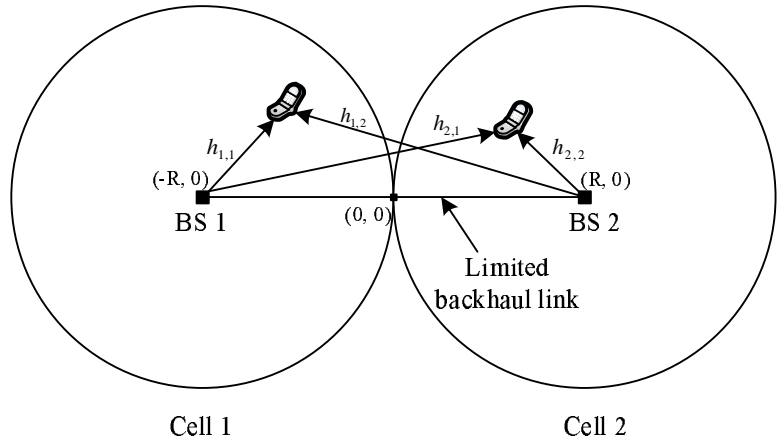}
\caption{A two-cell network. Each BS is serving a home user, which is suffering OCI from the neighboring BS.}\label{fig:2cell}
\end{figure}

\begin{figure}
\centering
\includegraphics[width=4in]{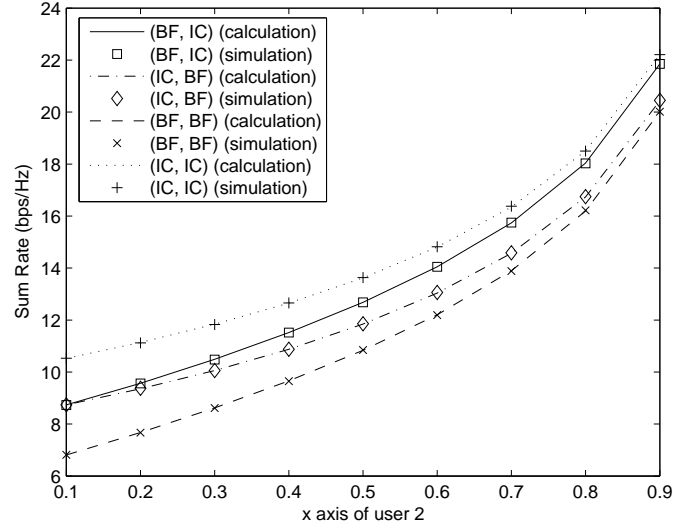}
\caption{Simulation and calculation results for different transmission strategy pairs. User 1 is at the cell edge $(-0.1R,0)$, and user  2 is moving from the cell edge to cell interior, $P_0=10$ dB, $\alpha=3.7$,
$N_t=4$.}\label{fig:simvscal_10dB}
\end{figure}

\begin{figure*}
\centering \subfigure[$P_0=-5$
dB]{\includegraphics[scale=.35]{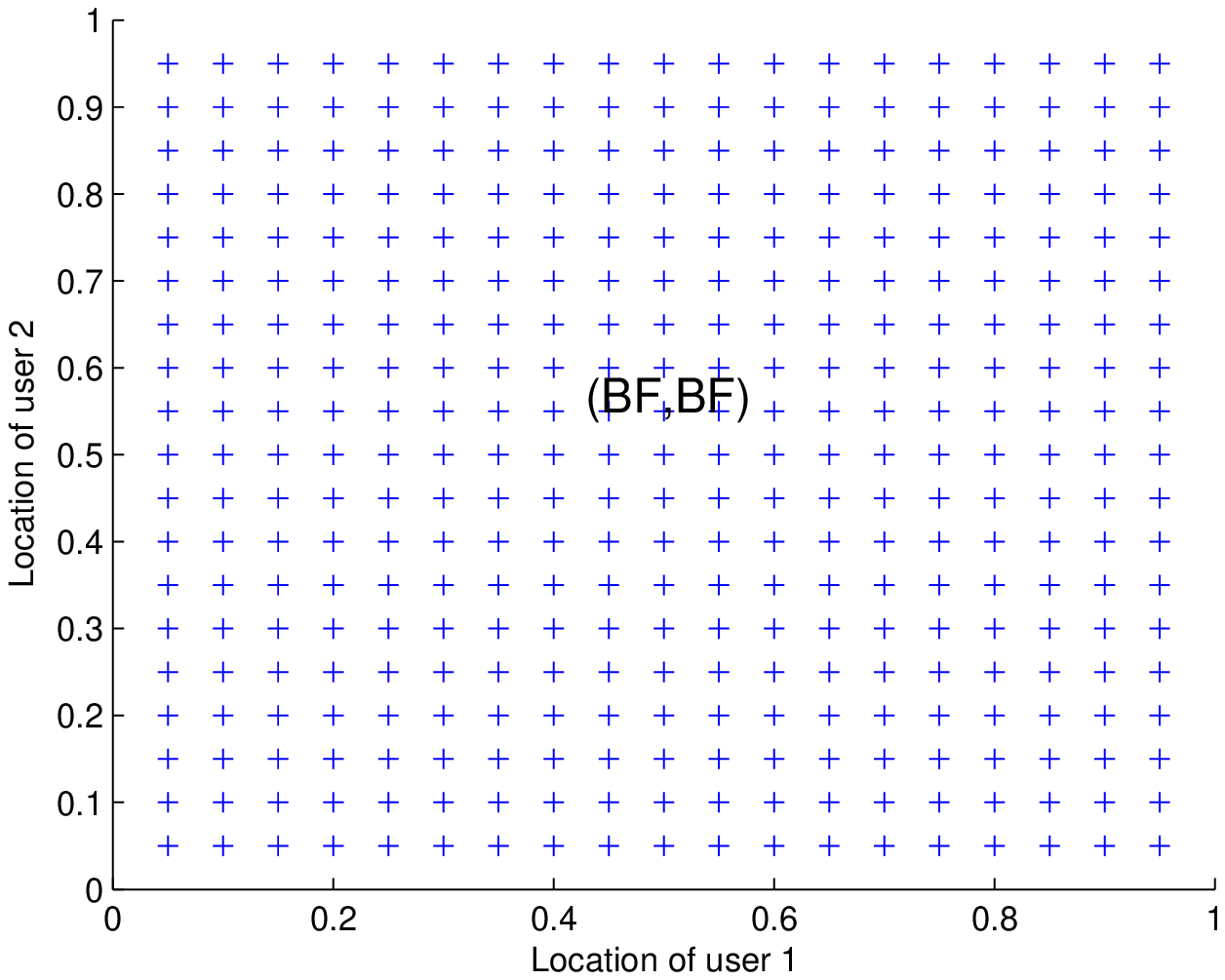}
\label{fig:Nt4_m5dB}}\hfill \subfigure[$P_0=5$
dB]{\includegraphics[scale=.35]{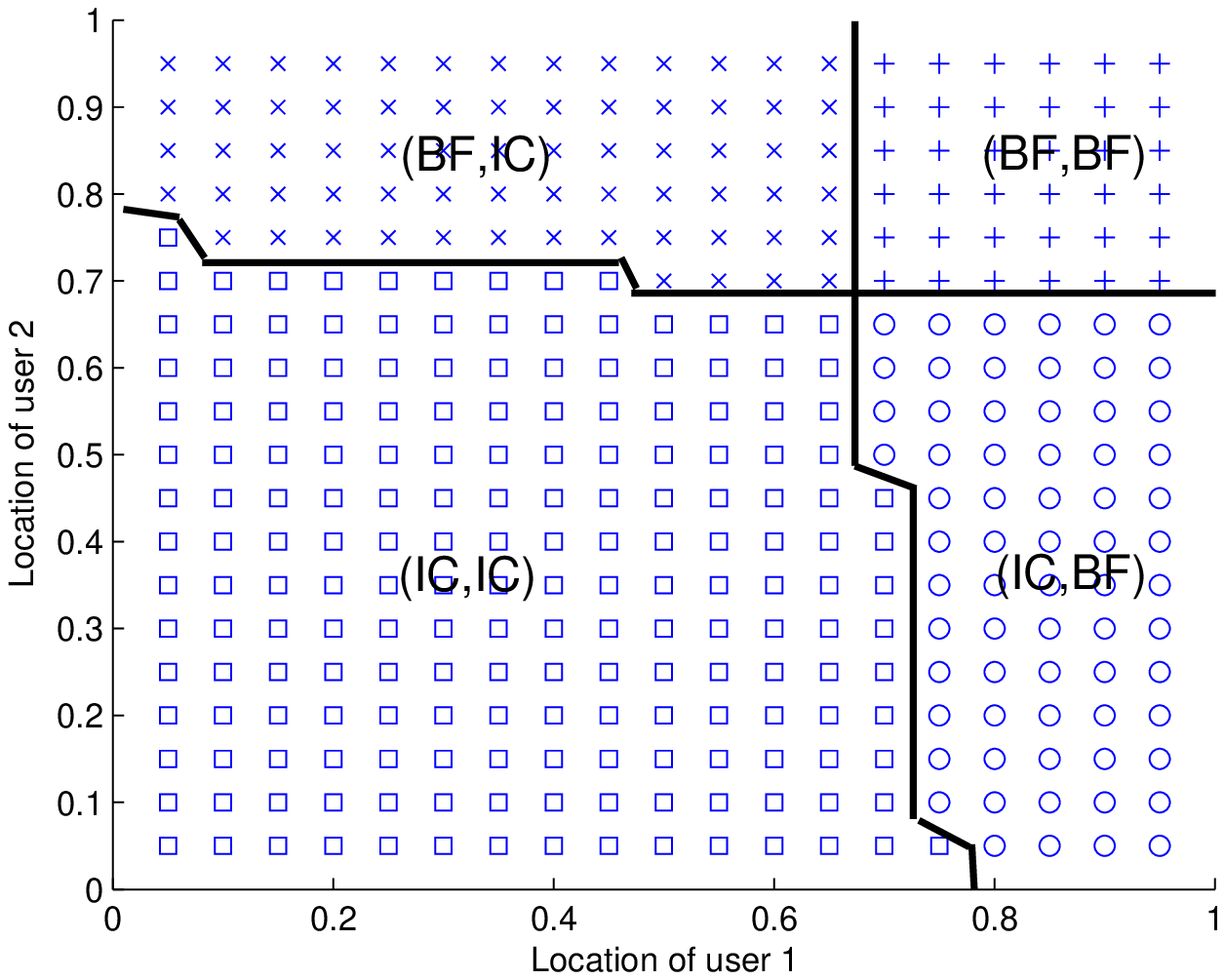} \label{fig:Nt4_5dB}}\hfill \subfigure[$P_0=10$
dB]{\includegraphics[scale=.35]{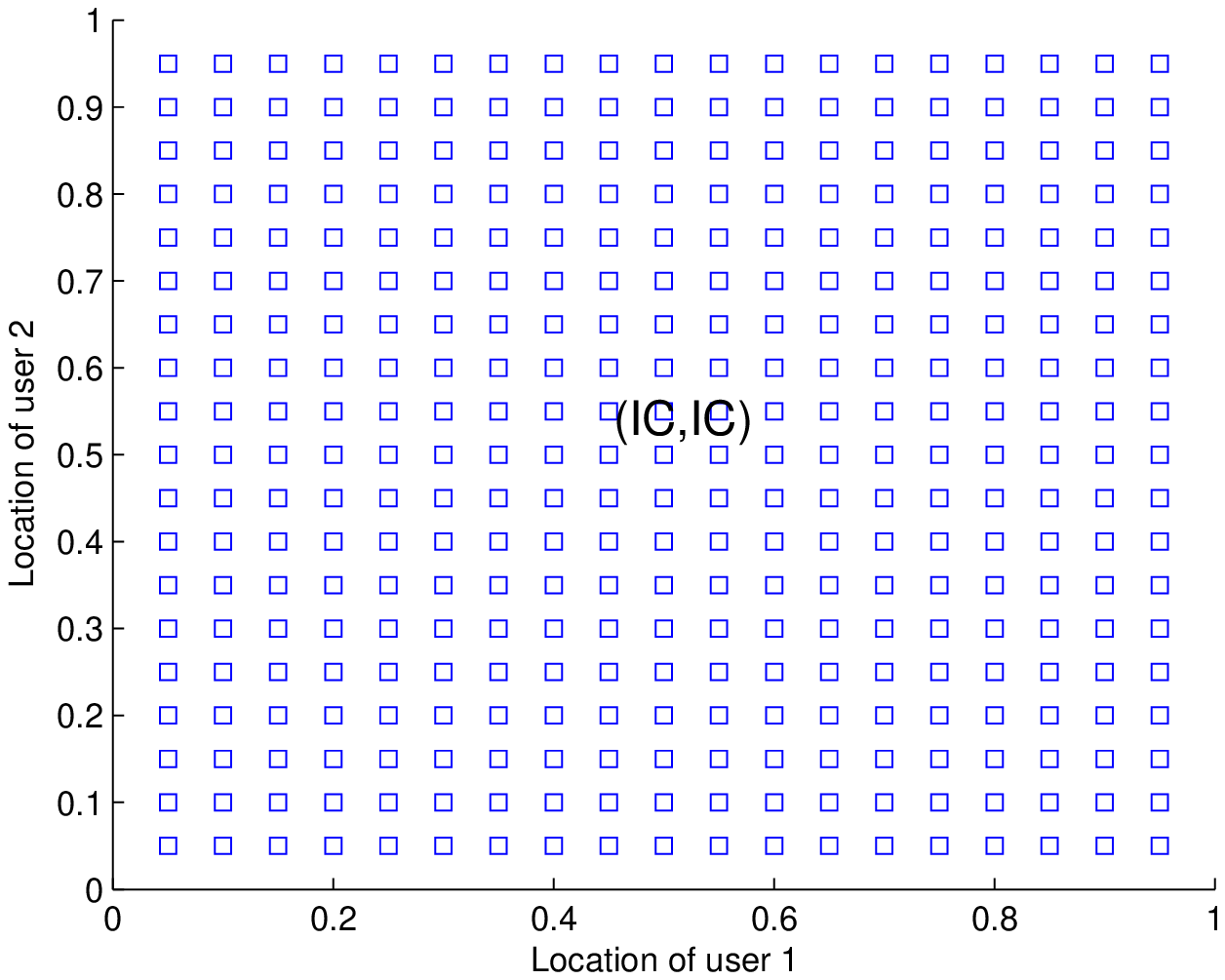}
\label{fig:Nt4_10dB}} \caption{Selected transmission strategies for different user locations in a 2-cell network, where $\alpha=3.7$, $N_t=4$, user 1 and user 2 are on the line connecting BS 1 and BS 2. The mark `x' denotes $(s_1,s_2)=(BF,IC)$, 'o' denotes $(s_1,s_2)=(IC,BF)$, '+' denotes $(s_1,s_2)=(BF,BF)$, and '$\Box$' denotes $(s_1,s_2)=(IC,IC)$.}
\label{fig:CSIT}
\end{figure*}

\begin{figure}
\centering
\includegraphics[width=4in]{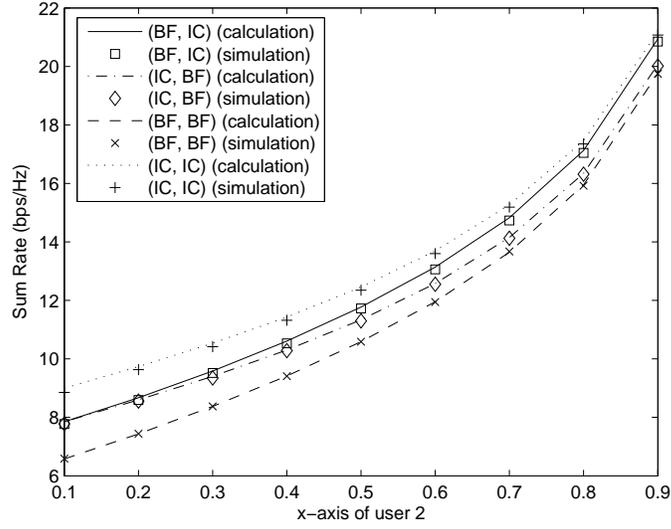}
\caption{Simulations and approximations for the sum throughput in a 2-cell network. User 2 is at $(-.1R,0)$, $P_0=10$ dB, $\alpha=3.7$, $B=10$, and $N_t=4$.}\label{fig:SimvsApprox_B10dB10}
\end{figure}

\begin{figure}
\centering
\includegraphics[clip=true,scale=.7]{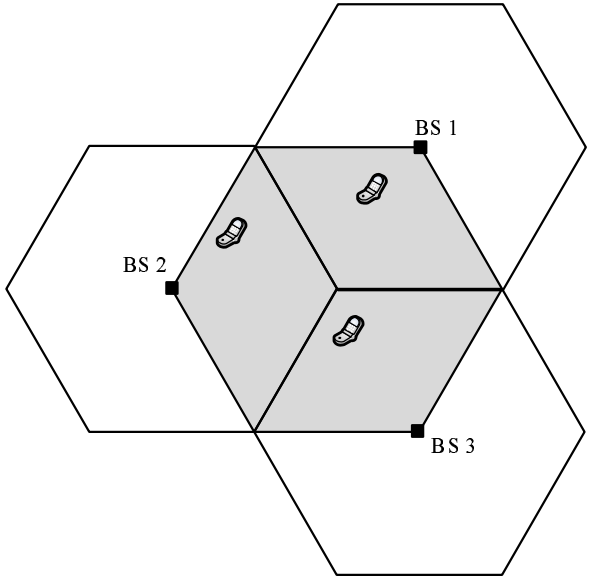}
\caption{A 3-cell network. The shadow area is considered as the ``inner area'', where users suffer high OCI from neighboring cells.}\label{fig:3cell}
\end{figure}

\begin{figure*}
\centering {\subfigure[Average
throughput per cell]{\includegraphics[width=3.1in]{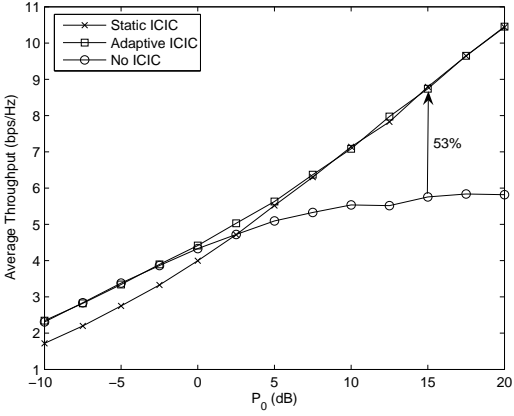} \label{fig:Ravg_CSIT}}
\hfil \subfigure[5th percentile
throughput]{\includegraphics[width=3.1in]{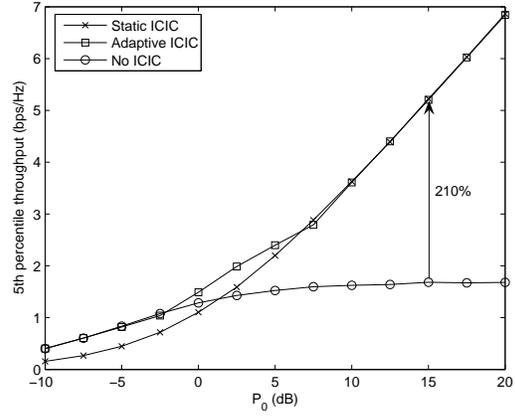}
\label{fig:R5_CSIT}}} \caption{Comparison of systems with different transmission
strategies in a 3-cell network with perfect CSI.} \label{fig:Comp}
\end{figure*}

\begin{figure}
\centering
\includegraphics[width=4in]{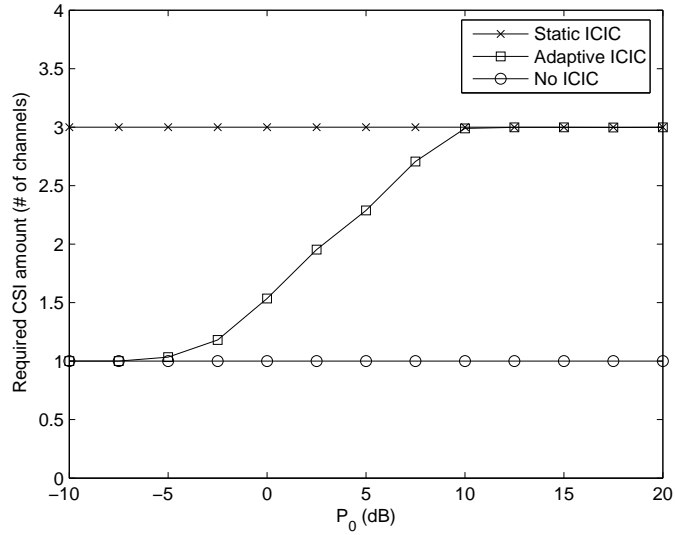}
\caption{The CSI requirements for different systems, which are expressed in number of channel directions.}\label{fig:FB}
\end{figure}

\begin{figure*}
\centering {\subfigure[Average
throughput per cell]{\includegraphics[width=4in]{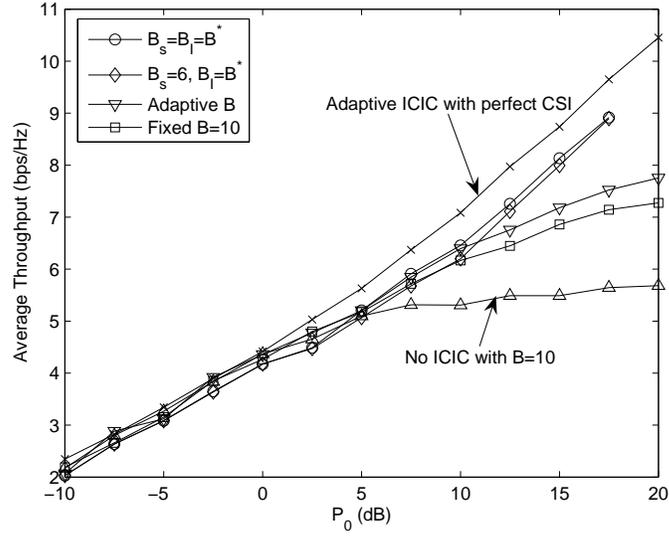} \label{fig:Ravg_LFB}}
\hfil \subfigure[5th percentile
throughput]{\includegraphics[width=4in]{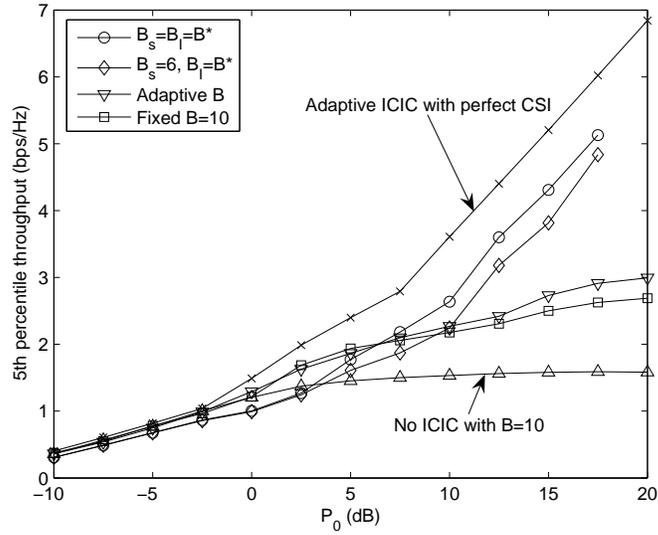}
\label{fig:R5_LFB}}} \caption{Comparison of systems with different
feedback strategies in a 3-cell network, where $B_s$ and $B_I$ are the numbers of feedback bits for the home BS and the helper BS, respectively, and $B^\star$ is given in \eqref{eq:scaleB} for different $P_0$.} \label{fig:Comp}
\end{figure*}

\end{document}